\documentclass[twoside,leqno,twocolumn]{article}

\usepackage[letterpaper]{geometry}

\usepackage{ltexpprt}

\usepackage[table]{xcolor}

\usepackage{amssymb}

\usepackage{mathtools}

\usepackage[linktocpage]{hyperref}

\usepackage[noabbrev,capitalise]{cleveref}

\usepackage{booktabs}

\newcommand{\deff}[1]{\textbf{#1}} %

\newcommand{\bR}{\mathbb{R}}

\newcommand{\flag}{\mathrm{Clique}}

\newcommand{\crit}{\mathrm{crit}}

\providecommand\given{}
\newcommand\SetSymbol[1][]{%
  \nonscript\:#1\vert{}
  \allowbreak{}
  \nonscript\:
  \mathopen{}}
\DeclarePairedDelimiterX\Set[1]\{\}{%
  \renewcommand\given{\SetSymbol[]}
  #1
}

\title{Filtration-Domination in Bifiltered Graphs}

\author{Ángel Javier Alonso\thanks{Graz University of Technology, Graz, Austria.
    Supported by Austrian Science Fund (FWF) grant P 33765-N.}
  \and Michael Kerber\footnotemark[1]
\and Siddharth Pritam\thanks{Shiv Nadar University, Delhi NCR, India.}}

\date{}

\begin{document}

\maketitle

\begin{abstract}
  \small\baselineskip=9pt
Bifiltered graphs are a versatile tool for modelling relations
between data points across multiple grades of a two-dimensional scale.
They are especially popular in topological data analysis, where the
homological properties of the induced clique complexes are studied.
To reduce the large size of these clique complexes, we identify
\emph{filtration-dominated} edges of the graph, whose removal
preserves the relevant topological properties. We give two algorithms
to detect filtration-dominated edges in a bifiltered graph and analyze
their complexity.
These two algorithms work directly on the bifiltered graph, without first
extracting the clique complexes, which are generally much bigger.
We present extensive experimental evaluation which shows
that in most cases, more than 90\% of the edges can be removed.
In turn, we demonstrate that this often leads to a substantial speedup,
and reduction in the memory usage, of the computational pipeline
of multiparameter topological data analysis.
\end{abstract}

\section{Introduction}

\subparagraph{Motivation and problem statement.}
A \emph{bifiltered graph} is a finite simple graph $G=(V,E)$ together with a
function $f$ that assigns to each \emph{grade} $(s,t)\in\mathbb{R}^2$ a subgraph $G_{s,t}$ of $G$. Moreover, these subgraphs are nested, that is, whenever $s\leq s'$ and
$t\leq t'$, we have $G_{s,t}$ is a subgraph of $G_{s',t'}$. Such bifiltered
graphs appear naturally in the area of multiparameter persistence, a theme
within topological data analysis that has received increasing attraction
recently. The idea is that a bifiltered graph models the relations of data across
various grades of a two-dimensional scale.
See \cref{fig:circles} for an illustration of these concepts.

Having a bifiltered graph, the pipeline of multiparameter persistence usually proceeds by deriving a bifiltered simplicial complex (a higher-dimensional analogue of a graph) by considering the \emph{clique complexes} of the subgraphs
and to study their topological properties in terms of (persistent) homology.
We postpone an explanation of these topological terms, which are not central to the results of the paper, to Section~\ref{sec:topological_motivation};
the main problems in this pipeline are that the clique complexes derived from a graph can be much larger than the size of the graph, and computing its topological properties requires algorithms that, despite many ongoing efforts, scale not too favorably with respect to the complex size.

\begin{figure*}[htb] \centering
  \includegraphics[width=\textwidth]{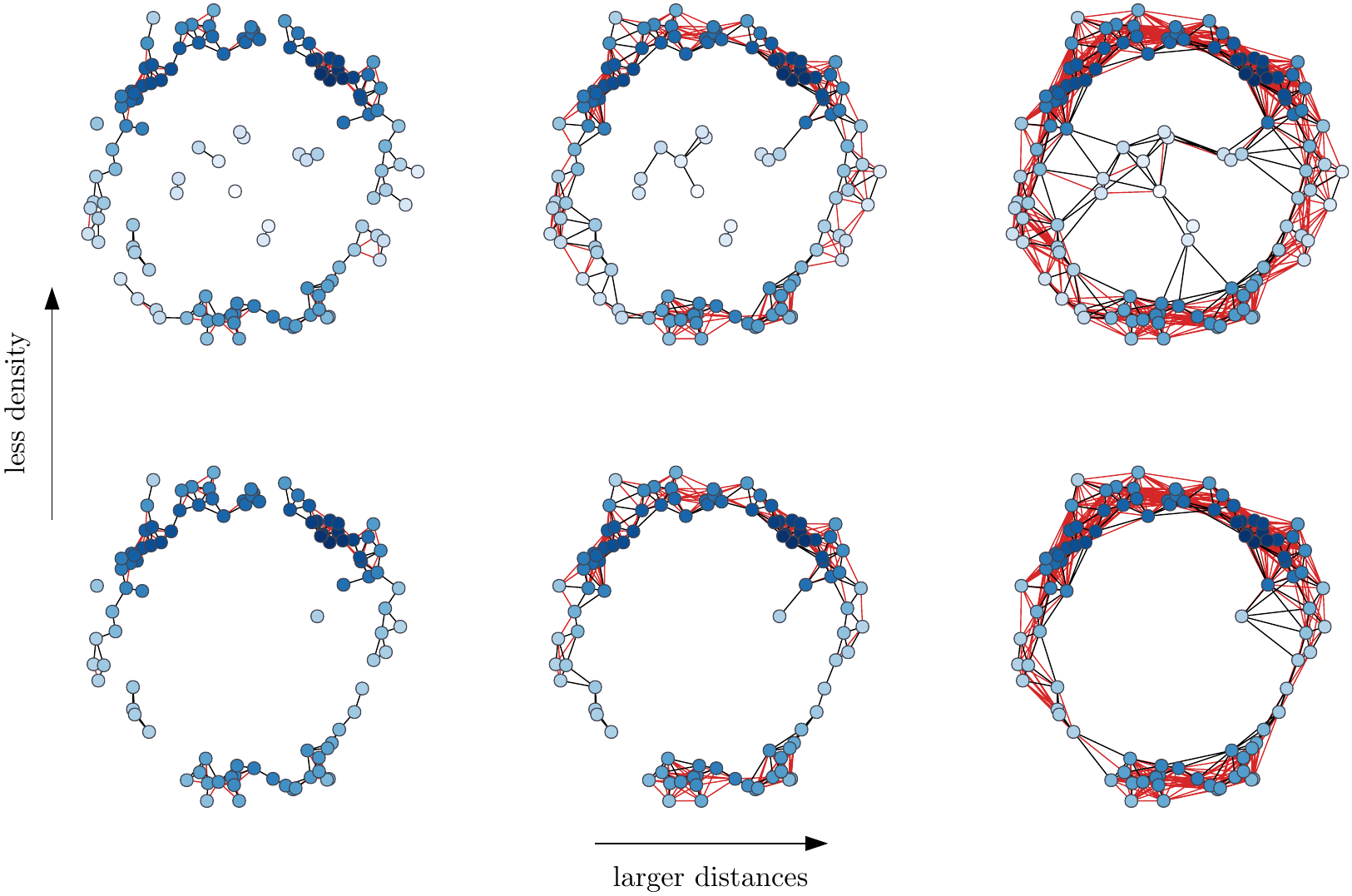}
  \caption{Illustration of a bifiltered graph. 
We sample points from a circle with noise and outliers.
We estimate a density of each point, encoded by the color (darker means denser).
Then, we define $G_{s,t}$ as the graph whose vertex set are points with density
at least $s$ and whose edges have Euclidean length at most $t$.
The figure shows six of these graphs with quite different topological configurations.
The edges colored red are filtration-dominated during the run of our
algorithm,
and can be removed without changing the topology of the graph, in a sense made precise
in~\cref{sec:topological_motivation}.}\label{fig:circles}
\end{figure*}

Therefore, we study the following problem: given a bifiltered graph $G$, our goal
is to compute another, hopefully much smaller bifiltered graph $G'$ such that
the clique complexes of $G$ and $G'$ are ``equivalent under the lens of homology''. Formally, we require that the $2$-parameter persistence modules (see Section~\ref{sec:topological_motivation}) induced by $G$ and $G'$ are isomorphic.

\subparagraph{Contribution.} Even though the formal problem statement requires topological notation, our contribution can be in large parts understood
in terms of elementary graph theory: Previous work~\cite{boissonnatEdgeCollapsePersistence} defined
the concept of \emph{dominated edges} in a graph which can be expressed and
checked combinatorially through the local neighborhood of the edge (see Section~\ref{sec:filtration_intro}). We extend this concept
to \emph{filtration-dominated edges} which are edges dominated
at all grades in a bifiltered graph. 
While filtration-domination is implicitly used in~\cite{boissonnatEdgeCollapsePersistence,glisseSwapShiftTrim2022},
our definition appears to be novel.
It is not difficult to verify, using well-known algebraic machinery, that removing filtration-dominated edges from a bifiltered graph preserves its homological properties (Theorem~\ref{thm:isomorphic}). This suggests a
greedy algorithm to decrease the size of a bifiltered graph: look for a
filtration-dominated edge, remove it, and repeat with the smaller graph.

We design and analyze efficient algorithms 
to search for filtration-dominated edges (Section~\ref{sec:algorithms}). 
We first give an algorithm for finding \emph{strongly filtration-dominated
  edges}; such edges have an extra condition compared to the non-strong definition
and can be detected efficiently in time proportional to the size of the local
neighborhood of an edge. We use this algorithm as a subroutine to
detect (non-strong) filtration-dominated edges using a simple
point-location data structure in the plane. We analyze this algorithm and show
that its runtime increases by a factor of $O(k\log k)$ compared to checking
strong filtration-domination, where $k$ is the maximal degree in the graph.

Our final contribution is an extensive experimental evaluation of our novel
concepts and algorithms on a broad collection of artificial and real-world data.
We show that in practice, only removing strongly filtration-dominated edges
yields a slightly larger output graph than removing all
filtration-dominated edges, but usually outperforms the general method in terms
of computation time. Moreover, our
greedy removal strategy can choose in what order to check edges for
(strong) filtration-domination. We show that the more removals happen when checking edges
that are added late in the bifiltered graph first. Finally, we demonstrate that
our approach speeds up the pipeline of multiparameter persistence: the usual
next step after constructing the clique complex is the computation of a
\emph{minimal presentation}~\cite{lesnickComputingMinimalPresentations2020}
which is a compact representation of the homological information of the complex.
We demonstrate that patching our algorithm before the highly optimized minimal
presentation algorithm \textsc{mpfree}~\cite{kerberFastMinimalPresentations2020}
yields speedups of more than an order of magnitude.

\subparagraph{Related work.}
Our approach is inspired by, and generalizes the line of research by Pritam et
al.~\cite{boissonnatEdgeCollapsePersistence,glisseSwapShiftTrim2022} who
considered the same problem for unifiltered graphs 
(i.e., the scale is unidimensional 
and the grades are real values). The extension to bifiltered graphs
is not entirely straightforward; most difficulties arise from the fact
that $\bR^2$ is only partially ordered. As an example,
in a unifiltered graph, checking for domination of an edge reduces to
examine the subgraphs of the grades where an edge is added; this is not true
for bifiltered graphs, because two edges can appear at incomparable grades
in $\bR^2$. 
On the other hand, our solution yields an algorithm also for the unifiltered
case by just ignoring a coordinate. Interestingly, our experiments show that our
algorithm for bifiltered graphs sometimes yields a smaller output than
the highly optimized, most recent algorithm for unifiltered
graphs~\cite{glisseSwapShiftTrim2022}, when run in the underlying unifiltered
graphs, as described in~\cref{sec:experiments}.

Persistent Homology has boosted the field of applied topology in the last
20 years. We refer to
textbooks~\cite{deyCompTop2022,edelsbrunnerCompTop2010,oudotPersistenceTheory2015} and
surveys~\cite{botnanIntroductionMultiparameterPersistence2022,carlssonPersHom2020}.
One of the most studied objects is the \emph{Vietoris-Rips filtration}, both in theory~\cite{aa-circle,chazalPersistenceStabilityGeometric2014,sheehyLinearSizeApproximationsVietoris2013}
and in
practice~\cite{Aktas2019,bleherTopologicalDataAnalysis2022,Lo2018}:
From a finite metric space we construct a complete weighted graph $G$ where the weight is simply
the distance of two points, resulting in a unifiltered graph depicted in the top row
of~\cref{fig:circles}.

The introduction of a second parameter is often motivated by the fact that
homological properties might change significantly in the presence of
outliers~\cite{buchetTopAnOutliers2015,carlssonTheoryMultidimensionalPersistence2009}. 
The natural idea is to combine the length parameter with some density measure
on the points~\cite{bmt-topological,cgos-jacm,rs-stable}.
The clique complexes of the corresponding bifiltered graph, which is depicted in \cref{fig:circles},
are called its \emph{density-Rips bifiltration}.
While we concentrate on this case in the experiments, there exist numerous
alternative constructions;
see~\cite{blumbergStability2ParameterPersistent2022,botnanIntroductionMultiparameterPersistence2022} for a comprehensive overview.

The restriction to exactly two parameters might appear rather specialized, as we
can easily extend the concept of bifiltered graphs to multifiltered graphs,
where the scale consists of grades in $\bR^d$, with $d$ fixed. This choice is made partially for
brevity and clearness of presentation; our approach can be generalized to more
parameters with some extra care that we discuss in the conclusion. Moreover, on
the computational side, the case of two parameters has been studied extensively
recently~\cite{dkm-computing,kn-efficient,kerberFastMinimalPresentations2020,lesnickInteractiveVisualization2D2015},
and our approach contributes to establish a solid algorithmic
layer for the case of bifiltered data sets.

As mentioned, the major obstacle in processing bifiltered simplicial complexes
is their sheer size. A common practice is to construct the simplicial complex
only up to a target dimension $p$
and/or up to a maximal grade. This offers a trade-off between the size
of the object to be processed and the information captured by it.
A further line of research aims for removing simplices in
bifiltered simplicial complexes in a topology-preserving way with techniques
from Discrete Morse
Theory~\cite{alliliReducingComplexesMultidimensional2017,alliliAcyclicPartialMatchings2019,scaramucciaComputingMultiparameterPersistent2020};
see also~\cite{fk-chunk}. In a similar spirit, \emph{minimal
  presentations}~\cite{fkr-compression,kerberFastMinimalPresentations2020,lesnickComputingMinimalPresentations2020}
reduce a simplicial bifiltration to an algebraic description that captures the
homological information (in a fixed dimension) in a minimal form. Common to
these techniques is that the bifiltered complex has to be \emph{expanded} before
compression, that is, all its simplices have to be enumerated. Our approach, on
the other hand, acts solely on the underlying graph and can thus be used as a
preprocessing step for all mentioned approaches in the case of clique complexes.

\section{Filtration-dominated edges}\label{sec:filtration_intro}
We use the following basic notions for a graph $G=(V,E)$:
Two vertices are \deff{adjacent} if there is an edge between them.
A \deff{subgraph} $G'=(V',E')$ of $G$ is a graph with $V'\subseteq V$
and $E'\subseteq E$. The \deff{induced subgraph} of $V'\subseteq V$
is the subgraph $G'=(V',E')$ with the largest possible $E'$.
We write $G\setminus e$ for the subgraph of $G$ with the same vertices
and edges, except the edge $e$.
A \deff{$k$-clique} in a graph is a complete subgraph with $k$ vertices.
Finally, we call a vertex $v$ of $G$ \deff{dominating} if
it is adjacent to all other vertices in the graph.

Fixing a graph $G$ and an edge $e$ of $G$ with endpoints $u$ and $v$, 
we say that $w\in G$
is an \deff{(edge) neighbor} of $e$ if $w$ is both adjacent to $u$ and to $v$.
We say that $e$ is \deff{dominated} in $G$ by a vertex $w$ if $w$ is an edge neighbor of $e$, and every other edge neighbor of $e$ is adjacent to $w$.
Equivalently, writing $N_{G}(e)$ for the set of all edge neighbors,
$e$ is dominated if $w$ is a dominating vertex in the subgraph of $G$ induced by $N_G(e)$. Note that the vertex $w$ need not to be dominating in the whole graph $G$. Also, an edge might be dominated by more than one vertex. We say that $e$ is \deff{dominated} in $G$ if it is dominated by some vertex
in $G$. 

Recall the definition of a bifiltered graph from the beginning of the Introduction. In the following, when we talk about a bifiltered graph, 
we suppress the underlying function $f$ from the notation and refer to the subgraph
at a grade $(s,t)$ as $G_{s,t}$. Note that a subgraph $G'$ of a bifiltered graph $G$ is canonically bifiltered as well by defining $G'_{s,t}:=G_{s,t}\cap G'$ (where the intersection is taken vertex- and edge-wise).

The following simple definitions is the main concept of this work:
\begin{Definition}\label{def:filtration_domination}
  Let $G$ be a bifiltered graph. An edge $e$ is
  \deff{filtration-dominated} in $G$ if for every $(s,t)\in\mathbb{R}^2$ for
  which $e$ is in $G_{s,t}$, the edge $e$ is also dominated in $G_{s,t}$.

  We say that $e$ is \deff{strongly filtration-dominated} by $v\in G$,
if for every $(s,t)\in\mathbb{R}^2$ for which $e$ is in $G_{s,t}$, the edge $e$ is
dominated by $v$ in $G_{s,t}$.
\end{Definition}
Note that if $e$ is strongly filtration-dominated (by some vertex),
it is filtration-dominated, but the converse might not hold.

The idea behind the definition is that filtration-dominated edges may be removed from the bifiltered graph without changing its relevant topological properties~-- we postpone the precise discussion to Section~\ref{sec:topological_motivation}. Assuming this fact for now, we obtain a simple greedy framework
to compress a bifiltered graph: we traverse the edges in arbitrary order, and,
if an edge $e$ is filtration-dominated, remove it from $G$, and continue
considering $G\setminus e$.

Instead of deciding whether $e$ is filtration-dominated in our framework,
we can use any other predicate as long as we guarantee that only
filtration-dominated edges are removed. In particular, we can check
whether $e$ is strongly filtration-dominated by some vertex of the graph.
This change might result in keeping some filtration-dominated edges in the output; however, as we explain in Section~\ref{sec:algorithms},
strong filtration-domination can be checked faster than filtration-domination,
so there is a trade-off between compression rate and runtime.

The order in which the greedy algorithm traverses the edges has a significant
effect on the runtime of the algorithm. As we show in Section~\ref{sec:experiments}, it is beneficial to check edges first that appear late in the bifiltered graph.

\section{Topological motivation}\label{sec:topological_motivation} 
The purpose of this section is to justify the notion of filtration-domination
from the previous section. We start with the main theorem of the section.
Its proof is relatively short and follows straightforwardly from
previous work; the bulk of this section is devoted to explain the terms
used in the theorem. While these definitions are mostly self-contained,
some machinery from (basic) algebraic topology is unavoidable.

\begin{theorem}\label{thm:isomorphic}
Let $G$ be a bifiltered graph with a filtration-dominated edge $e$,
and let $G':=G\setminus e$. Let $\flag(G)$ denote the bifiltered
clique complex induced by $G$, and $\mathcal{H}_p(\flag(G))$ denote the $2$-parameter persistence module induced by that clique complex in dimension $p$,
where $p\geq 0$ is an arbitrary integer and homology is taken over any finite field. Define $\mathcal{H}_p(\flag(G'))$ in the same way. Then,
the two defined persistence modules are isomorphic.
\end{theorem}

\subparagraph{Simplicial complexes.} A \deff{simplicial complex} $K$ is
a collection of subsets of a non-empty finite set, $V$, its \deff{vertex
  set}, that is closed under taking subsets. That is, for every $A$ in
$K$, all the subsets of $A$ are in $K$. An element of $K$ of cardinality $k+1$ is called a \deff{$k$-simplex}, and $k$ is called the \deff{dimension}
of the simplex.
A \deff{subcomplex} $L$ of $K$ is a subcollection that is a
simplicial complex itself. 
Note that a collection of $1$-simplices $E$ over $V$ simply defines a graph
$G=(V,E)$. Hence, simplicial complexes generalize graphs
to higher dimensions. 

The only type of simplicial complex that we consider in this work are
\deff{clique complexes} (also called ``flag complexes''):
given a graph $G=(V,E)$, let $\flag(G)$ be the simplicial complex
whose $k$-simplices are the $(k+1)$-cliques of $G$;
this is a simplicial complex because subsets of cliques are cliques,
and $\flag(G)$ contains $G$ because the $1$- and $2$-cliques of a graph are
its vertices and edges, respectively. 
We point out that the clique complex is generally a very large object:
in the extremal case of a complete graph $K_n$ with $n$ vertices,
$\flag(K_n)$ has $\binom{n}{k+1}$ $k$-simplices 
and hence a total number of $2^n$ simplices.

A \deff{bifiltered (simplicial) complex} is a simplicial complex $K$
with a function that assigns to each grade $(s,t)$ in $\bR^{2}$ a subcomplex $K_{s,t}$,
with the property that if $s\leq s'$ and $t\leq t'$, we have
$K_{s,t}$ is a subcomplex of $K_{s',t'}$. This definition generalizes
the notion of a bifiltered graph in the natural way. Given
a bifiltered graph $G$, its clique complex $\flag(G)$ is naturally bifiltered
by $\flag(G)_{s,t}:=\flag(G_{s,t})$, and we call it its \deff{clique bifiltration}.

\subparagraph{Persistence modules.}
From now on, we write for two grades $u = (u_1, u_2)$ and $v = (v_1, v_2)$ in $\bR^{2}$
that $u\leq v$ if  $u_1 \leq v_1$ and $u_2 \leq v_2$. Note that this
is a partial order on $\bR^2$.

A \deff{bigraded (or 2-parameter) persistence module} $M$ is a family of vector spaces $\Set{M_u}_{u\in\bR^{2}}$ together with linear maps $m_{u\to v} : M_{u} \to M_{v}$
whenever $u\leq v$, such that $m_{u\to u}$ is the identity function
and $m_{v\to w}\circ m_{u\to v}=m_{u\to w}$ whenever $u\leq v\leq w$.

The natural way to obtain a bigraded persistence module is to apply \emph{(simplicial) homology} over a fixed base field on a bifiltered simplicial complex.
We omit a formal definition of homology and only describe the idea: 
A simplicial complex $K$  can be interpreted as a topological space by embedding
its vertex set in sufficiently high dimension and drawing its simplices
as convex hulls of the embedded vertices such that no unwanted intersections happen;
for a fixed integer $p>0$, the $p$-th homology group $H_p(K)$ is a vector space that captures, informally speaking, the $p$-dimensional hole structure
of the embedding of $K$. 
One can show that $H_p(K)$ is independent of how $K$ is embedded.

Crucially for us, homology is \deff{functorial} which means in simplified terms
that for complexes $K\subseteq L$, we obtain an induced linear map
$i_{K\to L}:H_p(K)\to H_p(L)$ with the properties that $i_{K\to K}$ 
is the identity map and for complexes $K\subseteq L\subseteq M$,
we have $i_{L\to M}\circ i_{K\to L}=i_{K\to M}$.
It follows that applying $p$-dimensional homology (over a fixed base
field) on a bifiltered simplicial complex $K$ yields
a bigraded persistence module which we denote by $\mathcal{H}_p(K)$. 
A persistence module captures the homological properties of a bifiltered
simplicial complex and allows to rank the prominence of each hole in the dataset,
depending on the range of grades on which each hole is present.
That makes persistence modules (uni-, bi-, or multigraded) a central
concept of topological data analysis.

Let $\mathcal{M}$ and $\mathcal{N}$ two persistence modules with linear maps $m_{\cdot\to\cdot}$ and $n_{\cdot\to\cdot}$, respectively.
We say that $\mathcal{M}$ and $\mathcal{N}$  are \textbf{isomorphic} if there is a collection of vector space isomorphisms $\phi_{u} : \mathcal{M}_{u} \to \mathcal{N}_{u}$ for each $u \in \mathbb{R}^2$, such that $\phi_{v} \circ m_{u\to v} = n_{u\to v} \circ \phi_{u}$, for all $u \leq v$.
That means that we can switch back and forth between two isomorphic
persistence modules freely: if the persistence modules are given by
the homology of two bifiltered simplicial complexes, it means that
these complexes are indistinguishable at all grades in terms of homology.

\subparagraph{Proof of Theorem~\ref{thm:isomorphic}.}
We sketch the proof of the theorem.
All involved techniques are elementary tools from combinatorial topology,
but we cannot explain them in detail for the sake of brevity.
The major tool is the following simple property
that connects domination of edges in a graph with the homological properties
of the clique complex:
\begin{lemma}\label{lem:incl_isomorphism}
Let $G$ be a graph in which $e$ is a dominated edge. Then 
the inclusion $\flag(G\setminus e)\subset \flag(G)$ induces an isomorphism
$H_p(\flag(G\setminus e))\to H_p(\flag(G))$ of the homology groups
for every $p\geq 0$.
\end{lemma}
\begin{proof}
If $e$ is dominated, it means that the subgraph induced
by its edge neighbors has a dominating vertex $v$. 
In topological language, the induced subgraph is the \emph{link}
of the edge in $\flag(G)$, and this link is a \emph{simplicial cone}
with apex $v$. In this case, it is known that there is a sequence
of \emph{elementary simple collapses} that transforms $\flag(G)$
into $\flag(G\setminus e)$~\cite[Lemma 2.7]{Welker}\cite[Lemma 8]{AttaliLieutierExtendedCol}. Since such simple collapses
define a strong deformation retraction, the inclusion map
induces an isomorphism in homology for every $p\geq 0$.
\end{proof}

Now let $G$ be a bifiltered graph and let $e$ be filtration-dominated in $G$.
For every grade $u\in\bR^{2}$, the inclusion $(\flag(G\setminus e))_{u}\subseteq (\flag(G))_{u}$
induces an isomorphism in homology. If $e$ is not in $G$, this is trivial
since $G\setminus e=G$, and otherwise, it follows from the above lemma (Lemma~\ref{lem:incl_isomorphism}) because
$e$ is dominated by definition. Hence we have an isomorphism $\phi_u$ for every
grade $u$. Moreover, the fact that these isomorphisms are all induced
by inclusion maps is enough to show that these isomorphisms commute
with the linear maps of the persistence modules of $\flag(G)$
and $\flag(G\setminus e)$. This concludes the proof of Theorem~\ref{thm:isomorphic}.

\section{Algorithms}\label{sec:algorithms}
We explain next how to decide (strong) filtration-domination of an edge
in a bifiltered graph. Recall that we write $(x_1,y_1)\leq (x_2,y_2)$
for points in $\bR^2$ if $x_1\leq x_2$ and $y_1\leq y_2$. Moreover, for two arbitrary points $(x_1,y_1)$ and $(x_2,y_2)$,
we define $(x_1,y_1)\ast (x_2,y_2):= (\max(x_1,x_2),\max(y_1,y_2))$ as their \deff{join}
which is equivalently defined as the smallest element $(a,b)$ with respect to $\leq$ such that $(x_1,y_1)\leq (a,b)$ and $(x_2,y_2)\leq (a,b)$.
We also consider the join of more than $2$ elements, which is defined
analogously by the coordinate-wise maximum.

\subparagraph{1-critical bifiltered graphs.}
We focus on bifiltered graphs $G=(V,E)$
with the following property: for every edge $e$, there is a unique \deff{critical grade}, denoted by $\crit(e)$, in $\bR^2$ such that $e\in G_{s,t}$ if and only if
$\crit(e)\leq (s,t)$. In other words, for every edge, there is a unique
grade on which the edge ``enters'' the bifiltered graph, and hence, such
bifiltered graphs are called \deff{$1$-critical}. This is a loss of generality
since not all bifiltered graphs of interest are $1$-critical.
Nevertheless, the case of $1$-critical bifiltrations has received attention
in algorithmic contexts (e.g.,~\cite{carlssonCompMultiOneCrit2010,dx-generalized,kerberFastMinimalPresentations2020})
because of its simplicity.
Moreover, we restrict to the $1$-critical case solely for the sake of brevity,
and our approach extends to more general scenarios;
we comment further on this in the conclusion.

Fix a total order on the vertex set $V$ of the graph.
Our input is a $1$-critical bifiltered graph
$G=(V,E)$ in adjacency list representation, that is, we store an array
of linked lists, one for each vertex $v$.
This list consists of pairs $(w,\alpha)$, where $w$ is a vertex adjacent to $v$,
and $\alpha=\crit(\{v,w\})$. We assume the list is sorted by the first
entry in the pair, with respect to the fixed total order on $V$.

This data structure indeed determines a unique bifiltered graph: 
for instance, to compute
$G_{s,t}$, we can iterate through all adjacency lists and select exactly those
edges $e$ for which $\crit(e)\leq (s,t)$. Note that our representation does not
specify critical grades for vertices, and we just assume that every vertex
is present at all grades; this is not important since isolated vertices
do not affect whether an edge is filtration-dominated.

\subparagraph{Strong filtration-domination.}
Fix an edge $e=\{a,b\}$ and a vertex $v$ different from $a$ and $b$.
How can we check whether $v$ strongly filtration-dominates $e$? 
A necessary condition is that $v$ is an edge neighbor of $e$ 
at every grade where $e$ is present; in particular, $v$ must be
an edge neighbor of $e$ at $\crit(e)$. It is easy to verify that this is the case
if and only if $\crit(\{a,v\})\leq\crit(e)$ and $\crit(\{b,v\})\leq\crit(e)$.
We call $v$ a \deff{potential strong dominator} of $e$ in this case.

Furthermore, whenever the edge $e$ acquires an edge neighbor $w\neq v$ in the bifiltered graph, the edge $\{v,w\}$ has to be present in the graph.
This can also be checked efficiently, observing that
$w$ becomes edge neighbor of $e$ at the unique grade
$\crit(\{a,w\})\ast\crit(\{b,w\})\ast\crit(e)$, which we will denote by
$\crit_{e}(w)$, and it only needs
to be checked whether $\crit(\{v,w\})\leq \crit_{e}(w)$.

The last condition is also sufficient for strong filtration-domination
because it ensures that $v$ is connected to all edge neighbors at every grade. This suggests the following algorithm for deciding whether $e$ is strongly
filtration-dominated by some vertex: first determine the potential strong dominators
by iterating over the adjacency lists of $a$ and $b$ once, identifying
the common neighbors, and checking the above criterion for potential strong domination.
Then, we scan the adjacency lists of $a$, $b$, and all potential strong dominators
simultaneously in order, and check for every edge neighbor $w$ of $e$
and every potential strong dominator $v$ the second condition from above.
Because all adjacency lists are sorted according to the same total order,
this algorithm requires only one scan through every adjacency list
of a potential strong dominator. This implies a total running time to check for the
strong filtration-domination
of a fixed edge $e$ of
\[O(\deg(a)+\deg(b)+\sum_{v\in N_G(e)} \deg (v))\]
which is both bounded by $O(|E|)$ and $O(k^2)$, where $k$
is the maximal degree in the graph.

\subparagraph{Regions of non-domination.}
In order to decide filtration-domination efficiently, we extend the above
algorithm for strong filtration-domination in the following way: for
a fixed edge $e=\{a,b\}$ and a vertex $v$, we compute a data structure that can
answer the following query efficiently: Given a grade $(s,t)$,
is $e$ dominated by $v$ in the graph $G_{s,t}$?

For notational convenience, define for two grades $p,q\in\bR^2$
\[
\Delta(p,q):=\{r\in\bR^2:p\leq r\wedge q\not\leq r\}.
\]
The set $\Delta(p,q)$ can be visualized as the difference between two
upper-right quadrants in the plane, one anchored at $p$ and
one at $q$.
Note that if $q\leq p$, $\Delta(p,q)=\varnothing$. Generally, $\Delta(p,q)$
is the union of two stripes, one \deff{horizontal} and one \deff{vertical}
(where one or both stripes can be empty). See~\cref{fig:regions}, which
illustrates this and the next paragraph.

Reviewing the algorithm for the strong case, there are two reasons why
$v$ does not dominate $e$ in $G_{s,t}$: Firstly, $v$ might not be an edge neighbor of $e$ in $G_{s,t}$ which happens in the region $\Delta(\crit(e),\crit(\{a,v\})\ast\crit(\{b,v\}))$.
Secondly, we might have an edge neighbor $w$ of $e$ that is not adjacent to $v$, which happens in the region
$\Delta(\crit_{e}(w), \crit(\{v, w\})) = \Delta(\crit(e)\ast\crit(\{a,w\})\ast\crit(\{b,w\}), \crit(\{v,w\}))$.
The union of all $\Delta$-regions is precisely the set of grades $(s,t)$
for which $v$ does not dominate $e$ in $G_{s,t}$, and we call it the
\deff{region of non-domination} of $v$.

\begin{figure}[htbp] \centering
  \includegraphics[width=\linewidth]{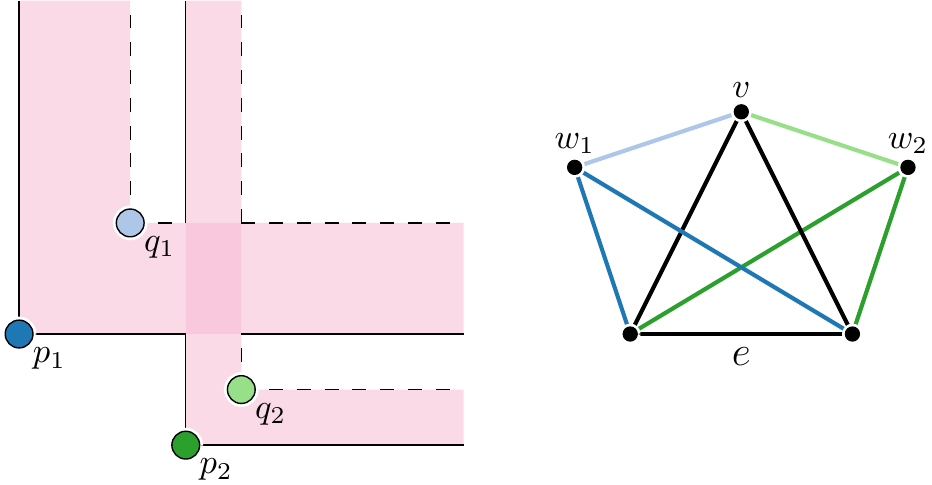}
  \caption{On the left, region of non-domination of the vertex $v$ when
    checking $e$ for filtration-domination on the graph of the right. The edge
    neighborhood of $e$ is $\Set{w_{1}, w_{2}, v}$. The grade in which $w_{1}$
    enters the edge neighborhood of $e$ is labeled $p_{1}$, and the
    critical grade of the edge between $v$ and $w_{1}$ is labeled $q_{1}$.
    Analogously for $w_{2}$ and $p_{2}$, and $q_{2}$. Then $v$ does not dominate $e$ in the
    region shaded in light pink, which is precisely
    $\Delta(p_{1}, q_{1})\cup\Delta(p_{2}, q_{2})$.}\label{fig:regions}
\end{figure}

Clearly,
in the same running time as for checking strong filtration-domination, we can
compute all $\Delta$-regions for all edge neighbors of $e$. Let $r$ denote the
number of edge neighbors of $e$. Then, $r+1$ is an upper bound for the number of $\Delta$-regions for any $v$. Given a query $(s,t)\in\bR^{2}$, we could now answer whether $v$ dominates $e$ in $G_{s,t}$ in $O(r)$ time, just
by checking whether $(s,t)$ is contained in any of the $\Delta$-regions.

We can reduce the query time to
$O(\log r)$ with $O(r\log r)$ preprocessing and $O(r)$ space, as we now briefly
describe. Recall that every $\Delta$-region is the union of a vertical stripe
and a horizontal stripe. Each vertical stripe is represented by a half-open
interval determining its $x$-range, and the value of the $y$-coordinate of the
segment that bounds it from below. Sweeping through all the interval endpoints
from left to right, we can easily obtain in $O(r\log r)$ time an ordered
sequence of $O(r)$ interior-disjoint vertical stripes whose union is precisely the union
of the original vertical stripes. Given a query grade in $\bR^{2}$, we can check
if it is contained in this disjoint union of vertical stripes in $O(\log r)$ time
via binary search. An analogous procedure can be used for the horizontal stripes,
yielding our desired bounds.

\subparagraph{Filtration-domination.}
The idea to decide filtration-domination is to use the regions of non-domination described above to check whether at least one edge neighbor
dominates $e$ at every grade in $\bR^{2}$. Naturally, as we cannot check all
grades,
we need to reduce the number of queries to a finite set.
It is necessary, but not sufficient,
to check for domination of $e$ at the grade $\crit_e(w)$ only;
however, it suffices to check at all \emph{joins} of such critical grades:
\begin{lemma}\label{lem:crit_joins}
With the notation as before, define
\[
C:=\Set{\crit(e)}\cup\Set{\crit_e(w_1)\ast\crit_e(w_2)\given w_1,w_2\in N_G(e)}
\]
(note that $w_1=w_2$ is allowed). Then $e$ is dominated for every grade
in $C$ if and only if it is filtration-dominated.
\end{lemma}
\begin{proof}
  If $e$ is filtration-dominated, it is dominated at every grade by definition.
  For the other direction, fix some grade $q$. Let $V_q$ denote the edge
  neighbors of $e$ in $G_q$; these are all edge neighbors $w$ of $e$ in $G$ such
  that $\crit_e(w)\leq q$. The crucial property is now that there exists some
  grade $c\in C$ with $c\leq q$ and $V_c=V_q$. To see that, note that we can
  move $q$ horizontally to the left without changing $V_q$ until we either hit
  the $x$-coordinate of some $\crit_e(w)$, or of $\crit(e)$. Then, we can move
  $q$ vertically down until one of the same events happens. After these two
  transformations, we end up at a join of the critical grades of two edge
  neighbors, or at the critical grade of $e$ itself.

  With $c$ as above, we define $G^{(e)}_c$ as the subgraph of $G_c$ induced by
  $V_c$, and likewise $G^{(e)}_q$ as the subgraph of $G_q$ induced by $V_q$.
  Since $V_c=V_q$, both graphs have the same vertex set, and since $G_c$ is a
  subgraph of $G_q$, $G^{(e)}_c$ is a subgraph of $G^{(e)}_q$. By assumption,
  $e$ is dominated in $G_c$, meaning that $G^{(e)}_c$ has a dominating vertex.
  Because a vertex in a graph remains dominating when adding edges to the graph,
  also $G^{(e)}_q$ has a dominating vertex, so $e$ is dominated in $G_q$.
\end{proof}

With that lemma, the algorithm for deciding filtration-domination of $e$ is as
follows: For every edge neighbor $w$ of $e$, prepare the data structure
to check for domination of $e$ by $w$ at a query grade, as described
in the previous paragraph. Then, compute the set $C$ as in the lemma, and query
each data structure for every point in $C$. Return whether for all $c\in C$, some
edge neighbor dominates $e$.

For the running time, let $r$ again denote the number of edge neighbors of $e$, and let $m$ be the number of edges.
We determine the $\Delta$-regions in $O(m)$ time and construct the $r$ 
data structures in total time $O(r^2\log r)$. The set $C$ consists of $O(r^2)$
elements, and for each of them, we need to query $r$ data structures, 
incurring a total cost of $O(r^3\log r)$. Hence, in total, we arrive at a runtime of
$O(m+r^3\log r)$ for deciding whether a fixed edge $e$ is filtration-dominated.

\section{Experimental Results}\label{sec:experiments}

\newcommand{\dataset}[1]{\texttt{#1}}

\subparagraph{Experimental setup.} We have implemented the algorithms in the
Rust programming language, and compiled them with Rust 1.60 and the highest
optimization levels. The code\footnote{
  \url{https://github.com/aj-alonso/filtration_domination}}, which includes all
code used to run the experiments and generate the tables, and all tested datasets are
publicly available. We ran all the experiments in a workstation with an Intel
Xeon E5-1650v3 CPU and 64 GB of RAM, running Ubuntu 20.04.4 LTS\@.

\subparagraph{Datasets.} We use real-world and synthetic
datasets. The real-world datasets are \dataset{netw-sc}, \dataset{senate},
\dataset{eleg}, \dataset{dragon}, and \dataset{hiv}, as collected and described
in~\cite{otterRoadmapComputationPersistent2017}. The synthetic datasets are
\dataset{sphere}, which is points sampled from the $2$-sphere in $\bR^{3}$ with outliers in the cube
$[-2, 2]^{3}$, as in~\cite{kerberFastMinimalPresentations2020},
\dataset{uniform}, points sampled uniformly at random in $[0,1]^{2}$,
\dataset{circle}, points sampled from a circle in the plane,
\dataset{torus}, points sampled from a torus in $\bR^{3}$, and
\dataset{swiss-roll},
a two-dimensional plane rolled up into a spiral in $\bR^{3}$ (see, for
example,~\cite{roweisNonlinearDimensionalityReduction2000}).

In each dataset, we assign a density for every point using 
the Gaussian kernel density estimation function~\cite{silvermanDensityEstimationStatistics1986}, 
with bandwidth parameter set to the 20th percentile of all distinct distances.
Then, we consider the complete graph over the point cloud,
bifiltered with respect to edge length and density, as depicted in \cref{fig:circles}.

\subparagraph{Comparing orders.}\label{sec:comparison_orders}
We check every edge for (strong) filtration-domination once in the algorithm, 
but the order is arbitrary. We investigate the effect of choosing different orders.
The grades are partially ordered by $\leq$, and there are four natural ways to
complete $\leq$ to a total order: lexicographic, colexicographic, 
reverse lexicographic and reverse colexicographic. In addition, we also tested a
random total order.

We run our algorithm once for each dataset and order. A comparison of the
results is shown in \cref{tab:order}. The reverse lexicographic and reverse
colexicographic result in a greater number of removed edges in all cases, and successfully
remove more than 90\% of the edges in almost all cases. 
This can be explained by the fact that late edges (with respect to $\leq$) in the
bifiltered graph tend to be
filtration-dominated, and removing them decreases the number of edge neighbors of earlier edges,
increasing their chances for being filtration-dominated. We also observe that more removals
naturally implies faster runs.
Note that for unifiltered graphs, recent work~\cite{glisseSwapShiftTrim2022}
also demonstrates the advantages of handling the edges in reverse order.

Among these two reverse orders, the reverse lexicographic
order is usually a bit better, and so it is the one we use in the rest of the
experiments.
\begin{table*}[!h]

\caption{\label{tab:order}Comparison of the edges removed when using different orders.
      For each dataset and order, we show the percentage of removed edges after a single run of the filtration-domination removal algorithm.
      The cases where the algorithm took more than 2 hours are marked with an ``---''.}
\centering
\resizebox{\linewidth}{!}{
\begin{tabular}[t]{lrrrrrrrrrr}
\toprule
\multicolumn{1}{c}{ } & \multicolumn{10}{c}{Datasets} \\
\cmidrule(l{3pt}r{3pt}){2-11}
Order & senate & eleg & netwsc & hiv & dragon & sphere & uniform & circle & torus & swiss roll\\
\midrule
\cellcolor{gray!6}{Random} & \cellcolor{gray!6}{48.2\%} & \cellcolor{gray!6}{55.4\%} & \cellcolor{gray!6}{27.2\%} & \cellcolor{gray!6}{---} & \cellcolor{gray!6}{---} & \cellcolor{gray!6}{23.0\%} & \cellcolor{gray!6}{45.2\%} & \cellcolor{gray!6}{9.6\%} & \cellcolor{gray!6}{54.6\%} & \cellcolor{gray!6}{60.0\%}\\
Colexicographic & 59.8\% & 87.0\% & 65.4\% & --- & --- & 21.2\% & 26.8\% & 17.4\% & 59.0\% & 50.8\%\\
\cellcolor{gray!6}{Lexicographic} & \cellcolor{gray!6}{61.2\%} & \cellcolor{gray!6}{92.6\%} & \cellcolor{gray!6}{96.6\%} & \cellcolor{gray!6}{---} & \cellcolor{gray!6}{---} & \cellcolor{gray!6}{21.2\%} & \cellcolor{gray!6}{68.8\%} & \cellcolor{gray!6}{17.4\%} & \cellcolor{gray!6}{72.8\%} & \cellcolor{gray!6}{65.6\%}\\
Reverse colex. & 90.4\% & 97.4\% & 99.4\% & 97.0\% & 97.6\% & \textbf{25.2\%} & 98.6\% & \textbf{27.4\%} & 92.2\% & 89.8\%\\
\cellcolor{gray!6}{Reverse lex.} & \cellcolor{gray!6}{\textbf{91.0\%}} & \cellcolor{gray!6}{\textbf{97.6\%}} & \cellcolor{gray!6}{\textbf{99.4\%}} & \cellcolor{gray!6}{\textbf{98.6\%}} & \cellcolor{gray!6}{\textbf{98.6\%}} & \cellcolor{gray!6}{25.0\%} & \cellcolor{gray!6}{\textbf{98.6\%}} & \cellcolor{gray!6}{24.2\%} & \cellcolor{gray!6}{\textbf{93.6\%}} & \cellcolor{gray!6}{\textbf{94.2\%}}\\
\bottomrule
\end{tabular}}
\end{table*}

\subparagraph{Performance.}\label{sec:comparison_removal}

We compare the filtration-domination and strong filtration-domination removal algorithms in~\cref{tab:removals}.
We observe that the number of remaining edges is smaller in the former variant as expected, but the ratio
is never much more than a factor of $2$, and sometimes close to $1$. We also see that the performance
of the latter variant is better (also expected), and the difference is sometimes more than an order of magnitude.
We conclude that the strong version seems the better choice in general, although the non-strong version might
be useful if subsequent computations scale very badly with the number of edges.

For comparison, we also run the state-of-the-art single-parameter
algorithm\footnote{As implemented in GUDHI~\cite{gudhi} version 3.6.0} of~\cite{glisseSwapShiftTrim2022} on the single-parameter
filtrations derived from our datasets (by dropping the density parameter).
Interestingly, we observe that our algorithm often returns smaller output graphs,
despite taking the density parameter into account and therefore being more selective
for removing edges. This raises the question whether our methodology might also help
to further improve the single-parameter case.

\begin{table*}

\caption{\label{tab:removals}Performance evaluation. The first two columns describe the datasets.
      Each group of columns contains three subcolumns: ``After'', the number of remaining edges after running the corresponding removal algorithm, ``\%'', the percentage of remaining edges, and ``Time (s)'', the time taken in seconds.}
\centering
\begin{tabular}[t]{lrrrrrrrrrr}
\toprule
\multicolumn{2}{c}{ } & \multicolumn{3}{c}{Filtration-domination} & \multicolumn{3}{c}{Strong filtration-domination} & \multicolumn{3}{c}{Single-parameter} \\
\cmidrule(l{3pt}r{3pt}){3-5} \cmidrule(l{3pt}r{3pt}){6-8} \cmidrule(l{3pt}r{3pt}){9-11}
Datasets & Before & After & \% & Time (s) & After & \% & Time (s) & After & \% & Time (s)\\
\midrule
\cellcolor{gray!6}{senate} & \cellcolor{gray!6}{5253} & \cellcolor{gray!6}{476} & \cellcolor{gray!6}{9.1\%} & \cellcolor{gray!6}{0.23} & \cellcolor{gray!6}{1101} & \cellcolor{gray!6}{21.0\%} & \cellcolor{gray!6}{0.03} & \cellcolor{gray!6}{242} & \cellcolor{gray!6}{4.6\%} & \cellcolor{gray!6}{0.01}\\
eleg & 43956 & 1026 & 2.3\% & 0.48 & 1254 & 2.9\% & 0.25 & 1224 & 2.8\% & 0.09\\
\cellcolor{gray!6}{netwsc} & \cellcolor{gray!6}{71631} & \cellcolor{gray!6}{424} & \cellcolor{gray!6}{0.6\%} & \cellcolor{gray!6}{0.42} & \cellcolor{gray!6}{426} & \cellcolor{gray!6}{0.6\%} & \cellcolor{gray!6}{0.37} & \cellcolor{gray!6}{476} & \cellcolor{gray!6}{0.7\%} & \cellcolor{gray!6}{0.13}\\
hiv & 591328 & 7820 & 1.3\% & 269.00 & 18942 & 3.2\% & 35.79 & 18648 & 3.2\% & 6.33\\
\cellcolor{gray!6}{dragon} & \cellcolor{gray!6}{1999000} & \cellcolor{gray!6}{29893} & \cellcolor{gray!6}{1.5\%} & \cellcolor{gray!6}{225.22} & \cellcolor{gray!6}{45514} & \cellcolor{gray!6}{2.3\%} & \cellcolor{gray!6}{66.79} & \cellcolor{gray!6}{53503} & \cellcolor{gray!6}{2.7\%} & \cellcolor{gray!6}{24.52}\\
\addlinespace
sphere & 4950 & 3710 & 74.9\% & 4.25 & 3800 & 76.8\% & 0.07 & 3752 & 75.8\% & 0.02\\
\cellcolor{gray!6}{uniform} & \cellcolor{gray!6}{79800} & \cellcolor{gray!6}{1123} & \cellcolor{gray!6}{1.4\%} & \cellcolor{gray!6}{0.99} & \cellcolor{gray!6}{1143} & \cellcolor{gray!6}{1.4\%} & \cellcolor{gray!6}{0.54} & \cellcolor{gray!6}{1138} & \cellcolor{gray!6}{1.4\%} & \cellcolor{gray!6}{0.21}\\
circle & 4950 & 3748 & 75.7\% & 4.76 & 4439 & 89.7\% & 0.08 & 4169 & 84.2\% & 0.03\\
\cellcolor{gray!6}{torus} & \cellcolor{gray!6}{19900} & \cellcolor{gray!6}{1280} & \cellcolor{gray!6}{6.4\%} & \cellcolor{gray!6}{0.30} & \cellcolor{gray!6}{1617} & \cellcolor{gray!6}{8.1\%} & \cellcolor{gray!6}{0.09} & \cellcolor{gray!6}{1618} & \cellcolor{gray!6}{8.1\%} & \cellcolor{gray!6}{0.04}\\
swiss-roll & 19900 & 1135 & 5.7\% & 0.37 & 1577 & 7.9\% & 0.09 & 2279 & 11.5\% & 0.04\\
\bottomrule
\end{tabular}
\end{table*}

\subparagraph{Multiple iterations.}
In the above experiments  we have run our algorithm only once
on each dataset. We can run the algorithm multiple 
times consecutively, by running the same algorithm on the output of the first run, and so on, possibly removing more edges each time. We test this assumption in this
set of experiments. For each dataset we have run the strong
filtration-domination removal algorithm 5 times, and the results are shown
in~\cref{tab:iterations}. The bulk
of the removals happens in the first iteration: the second and subsequent
iterations remove less than 3\% of the original edges, except in the
\texttt{senate} dataset, where the second iteration removes 6\% of the original
edges. We note that subsequent iterations take a fraction of the running time of
the first iteration, because they are run on smaller input.

\begin{table*}[!h]

\caption{\label{tab:iterations}Results after running the strong filtration-domination removal algorithm 5 consecutive times.
      There are 5 groups of columns, one for each iteration.
      The ``Removed'' column displays the percentage of the original edges removed in the corresponding iteration,
      and ``Time (s)'' displays the running time (in seconds) of the iteration.}
\centering
\resizebox{\linewidth}{!}{
\begin{tabular}[t]{lrrrrrrrrrr}
\toprule
\multicolumn{1}{c}{ } & \multicolumn{2}{c}{Iteration 1} & \multicolumn{2}{c}{Iteration 2} & \multicolumn{2}{c}{Iteration 3} & \multicolumn{2}{c}{Iteration 4} & \multicolumn{2}{c}{Iteration 5} \\
\cmidrule(l{3pt}r{3pt}){2-3} \cmidrule(l{3pt}r{3pt}){4-5} \cmidrule(l{3pt}r{3pt}){6-7} \cmidrule(l{3pt}r{3pt}){8-9} \cmidrule(l{3pt}r{3pt}){10-11}
Dataset & Time (s) & Removed & Time (s) & Removed & Time (s) & Removed & Time (s) & Removed & Time (s) & Removed\\
\midrule
\cellcolor{gray!6}{senate} & \cellcolor{gray!6}{0.03} & \cellcolor{gray!6}{79.0\%} & \cellcolor{gray!6}{0.00} & \cellcolor{gray!6}{6.0\%} & \cellcolor{gray!6}{0.00} & \cellcolor{gray!6}{2.8\%} & \cellcolor{gray!6}{0.00} & \cellcolor{gray!6}{0.8\%} & \cellcolor{gray!6}{0.01} & \cellcolor{gray!6}{0.2\%}\\
eleg & 0.26 & 97.2\% & 0.00 & 0.4\% & 0.01 & 0.0\% & 0.00 & 0.0\% & 0.00 & 0.0\%\\
\cellcolor{gray!6}{netwsc} & \cellcolor{gray!6}{0.37} & \cellcolor{gray!6}{99.4\%} & \cellcolor{gray!6}{0.00} & \cellcolor{gray!6}{0.0\%} & \cellcolor{gray!6}{0.00} & \cellcolor{gray!6}{0.0\%} & \cellcolor{gray!6}{0.00} & \cellcolor{gray!6}{0.0\%} & \cellcolor{gray!6}{0.00} & \cellcolor{gray!6}{0.0\%}\\
hiv & 35.98 & 96.8\% & 0.30 & 1.4\% & 0.12 & 0.6\% & 0.07 & 0.2\% & 0.06 & 0.2\%\\
\cellcolor{gray!6}{dragon} & \cellcolor{gray!6}{67.12} & \cellcolor{gray!6}{97.8\%} & \cellcolor{gray!6}{0.24} & \cellcolor{gray!6}{0.6\%} & \cellcolor{gray!6}{0.14} & \cellcolor{gray!6}{0.2\%} & \cellcolor{gray!6}{0.11} & \cellcolor{gray!6}{0.0\%} & \cellcolor{gray!6}{0.09} & \cellcolor{gray!6}{0.0\%}\\
\addlinespace
sphere & 0.06 & 23.2\% & 0.06 & 0.4\% & 0.05 & 0.0\% & 0.05 & 0.0\% & 0.05 & 0.0\%\\
\cellcolor{gray!6}{uniform} & \cellcolor{gray!6}{0.55} & \cellcolor{gray!6}{98.6\%} & \cellcolor{gray!6}{0.00} & \cellcolor{gray!6}{0.0\%} & \cellcolor{gray!6}{0.00} & \cellcolor{gray!6}{0.0\%} & \cellcolor{gray!6}{0.00} & \cellcolor{gray!6}{0.0\%} & \cellcolor{gray!6}{0.00} & \cellcolor{gray!6}{0.0\%}\\
circle & 0.07 & 10.4\% & 0.04 & 2.4\% & 0.04 & 1.2\% & 0.04 & 1.0\% & 0.04 & 1.6\%\\
\cellcolor{gray!6}{torus} & \cellcolor{gray!6}{0.09} & \cellcolor{gray!6}{91.8\%} & \cellcolor{gray!6}{0.00} & \cellcolor{gray!6}{1.6\%} & \cellcolor{gray!6}{0.01} & \cellcolor{gray!6}{0.4\%} & \cellcolor{gray!6}{0.00} & \cellcolor{gray!6}{0.0\%} & \cellcolor{gray!6}{0.00} & \cellcolor{gray!6}{0.0\%}\\
swiss-roll & 0.10 & 92.0\% & 0.00 & 2.6\% & 0.00 & 0.8\% & 0.00 & 0.2\% & 0.00 & 0.0\%\\
\bottomrule
\end{tabular}}
\end{table*}

\subparagraph{Different structure on the grades.}
Filtration-domination of an edge in a bifiltered graph depends on both the structure of the
underlying graph (the edge needs to be dominated), and the structure of the
grades (it should be dominated at each grade it is present). We now test how
the lack of structure on the grades affects the number of strongly
filtration-dominated edges we remove. We modify the densities of the described datasets
in two different ways: we zero out the
density parameter (effectively making the graph unifiltered),
and replace the densities by values sampled uniformly at
random. We then run the strong filtration-domination removal algorithm, and show
the results in~\cref{tab:random_densities}. For each case, we also count the
number of edges that are non-dominated in the subgraph of their respective critical grade:
these edges cannot be removed right away, and we refer to them by being ``free at birth''
in~\cref{tab:random_densities}. When we zero out the densities, in almost all
datasets there are less than 10\% of such edges. On the other hand, when using
random density values there are more than 80\% such edges in almost all datasets.
In a way, this measures how the lack of structure on the grades affects the
edges that can be removed, which is reflected upon the output: when using random
density values we remove less than 2\% of the edges in all
cases.

\begin{table*}[!h]

\caption{\label{tab:random_densities}Analysis of the removed edges under
      changes to the structure of the grades. There are three groups of columns:
      the first one represents the original dataset with no modification to the densities,
      in the second one we artificially zero out all the density values, and in the third one we
      replace the densities by random values sampled uniformly. ``Free at birth'' shows the
      percentage of edges that are not dominated when they appear (at their
      critical grade), and ``Removed'' is the percentage of edges removed after
      running our strong filtration-domination removal algorithm.}
\centering
\begin{tabular}[t]{lrrrrrr}
\toprule
\multicolumn{1}{c}{ } & \multicolumn{2}{c}{Original densities} & \multicolumn{2}{c}{Zeroed densities} & \multicolumn{2}{c}{Random densities} \\
\cmidrule(l{3pt}r{3pt}){2-3} \cmidrule(l{3pt}r{3pt}){4-5} \cmidrule(l{3pt}r{3pt}){6-7}
Dataset & Free at birth & Removed & Free at birth & Removed & Free at birth & Removed\\
\midrule
\cellcolor{gray!6}{senate} & \cellcolor{gray!6}{3.4\%} & \cellcolor{gray!6}{79.0\%} & \cellcolor{gray!6}{2.8\%} & \cellcolor{gray!6}{95.0\%} & \cellcolor{gray!6}{71.2\%} & \cellcolor{gray!6}{1.6\%}\\
eleg & 1.2\% & 97.2\% & 1.6\% & 97.0\% & 82.4\% & 1.4\%\\
\cellcolor{gray!6}{netwsc} & \cellcolor{gray!6}{0.2\%} & \cellcolor{gray!6}{99.4\%} & \cellcolor{gray!6}{0.2\%} & \cellcolor{gray!6}{99.4\%} & \cellcolor{gray!6}{84.0\%} & \cellcolor{gray!6}{1.2\%}\\
hiv & 0.6\% & 96.8\% & 1.4\% & 96.2\% & 92.0\% & 0.4\%\\
\cellcolor{gray!6}{dragon} & \cellcolor{gray!6}{0.4\%} & \cellcolor{gray!6}{97.8\%} & \cellcolor{gray!6}{0.6\%} & \cellcolor{gray!6}{97.2\%} & \cellcolor{gray!6}{94.6\%} & \cellcolor{gray!6}{0.2\%}\\
\addlinespace
sphere & 7.8\% & 23.2\% & 16.4\% & 23.8\% & 74.4\% & 0.6\%\\
\cellcolor{gray!6}{uniform} & \cellcolor{gray!6}{0.8\%} & \cellcolor{gray!6}{98.6\%} & \cellcolor{gray!6}{0.8\%} & \cellcolor{gray!6}{98.6\%} & \cellcolor{gray!6}{86.0\%} & \cellcolor{gray!6}{0.8\%}\\
circle & 6.4\% & 10.4\% & 10.0\% & 18.4\% & 67.8\% & 0.2\%\\
\cellcolor{gray!6}{torus} & \cellcolor{gray!6}{2.6\%} & \cellcolor{gray!6}{91.8\%} & \cellcolor{gray!6}{3.0\%} & \cellcolor{gray!6}{91.6\%} & \cellcolor{gray!6}{78.4\%} & \cellcolor{gray!6}{1.2\%}\\
swiss-roll & 2.6\% & 92.0\% & 4.0\% & 88.0\% & 79.4\% & 1.2\%\\
\bottomrule
\end{tabular}
\end{table*}

\subparagraph{Speeding up multiparameter persistent homology.}\label{sec:comparison_mpfree}
We evaluate the impact of our algorithm as a preprocessing step for computing
the minimal presentation of a persistence module in homology dimension $1$ induced by a bifiltered graph.
The standard approach is to first enumerate all triangles of the clique complexes
and then computing the minimal presentation through manipulations of the boundary matrices
of the simplicial complex; see~\cite{kerberFastMinimalPresentations2020,lesnickComputingMinimalPresentations2020}
for details. We suggest to first run our algorithm to obtain a smaller graph, and apply the two steps
above on that smaller graph instead.

For the computation of the minimal presentation we use 
\texttt{mpfree}\footnote{\url{https://bitbucket.org/mkerber/mpfree/src/master/}} (in parallel mode)
which is currently the fastest software available for this task.
As we see in \cref{tab:mpfree}, our preprocessing speeds up the computations
significantly in most examples.

Perhaps even more importantly, our approach also reduces
the memory consumption of the minimal presentation pipeline, which is the limiting factor
for large inputs, as can be seen in the case of
the \texttt{dragon} and \texttt{hiv} datasets, where the pipeline ran out of
memory without first doing preprocessing. In the datasets where the algorithm does not
remove many edges---\texttt{sphere} and \texttt{circle} in this case---the memory
savings are limited. However, in most cases the memory savings are high.
The most extreme case is the \texttt{uniform} dataset, where memory
consumption goes from 5.58 GB to 11.33 MB\@. This difference can be explained by
looking at the size of the input fed into \texttt{mpfree}: for first
dimensional homology, the input consists of all triangles and all edges of the
clique bifiltration. The number of 
edges is reduced from 79800 to 1143, as shown in \cref{tab:removals}, and the number of triangles is reduced from 10586800 to 875.

\begin{table*}[!h]

\caption{\label{tab:mpfree}Impact of our algorithm as a preprocessing step for minimal presentations.
      Inside each group of columns, the ``Build (s)'' column displays the time taken in seconds to build the clique bifiltration, ``mpfree (s)'' the time taken to run \texttt{mpfree},
      and ``Memory'' the maximum amount of memory used by the pipeline, over all the steps (including the preprocessing if applied).
      In addition, the ``Removal (s)'' column displays the time taken to run our algorithm, and ``Speedup'' is the speedup compared to not doing preprocessing. The $\infty$ symbol means
      that the pipeline ran out of memory, and in that
      case both the timing and speedup values are marked with an ``---''.}
\centering
\begin{tabular}[t]{lrrrrrrlr}
\toprule
\multicolumn{1}{c}{ } & \multicolumn{3}{c}{No preprocessing} & \multicolumn{4}{c}{With preprocessing} \\
\cmidrule(l{3pt}r{3pt}){2-4} \cmidrule(l{3pt}r{3pt}){5-8}
Dataset & Memory & Build (s) & mpfree (s) & Memory & Removal (s) & Build (s) & mpfree (s) & Speedup\\
\midrule
\cellcolor{gray!6}{senate} & \cellcolor{gray!6}{87.98 MB} & \cellcolor{gray!6}{0.05} & \cellcolor{gray!6}{0.49} & \cellcolor{gray!6}{8.2 MB} & \cellcolor{gray!6}{0.03} & \cellcolor{gray!6}{0.00} & \cellcolor{gray!6}{0.10} & \cellcolor{gray!6}{4.15}\\
eleg & 2.17 GB & 1.32 & 11.52 & 8.36 MB & 0.29 & 0.00 & 0.03 & 40.13\\
\cellcolor{gray!6}{netwsc} & \cellcolor{gray!6}{4.25 GB} & \cellcolor{gray!6}{3.03} & \cellcolor{gray!6}{21.05} & \cellcolor{gray!6}{11.42 MB} & \cellcolor{gray!6}{0.37} & \cellcolor{gray!6}{0.00} & \cellcolor{gray!6}{0.01} & \cellcolor{gray!6}{63.37}\\
hiv & $\infty$ & --- & --- & 366.25 MB & 35.74 & 0.16 & 2.25 & ---\\
\cellcolor{gray!6}{dragon} & \cellcolor{gray!6}{$\infty$} & \cellcolor{gray!6}{---} & \cellcolor{gray!6}{---} & \cellcolor{gray!6}{382.55 MB} & \cellcolor{gray!6}{67.10} & \cellcolor{gray!6}{0.22} & \cellcolor{gray!6}{4.57} & \cellcolor{gray!6}{---}\\
\addlinespace
sphere & 90.33 MB & 0.05 & 0.50 & 64.07 MB & 0.07 & 0.03 & 0.37 & 1.17\\
\cellcolor{gray!6}{uniform} & \cellcolor{gray!6}{5.58 GB} & \cellcolor{gray!6}{3.67} & \cellcolor{gray!6}{36.99} & \cellcolor{gray!6}{11.33 MB} & \cellcolor{gray!6}{0.56} & \cellcolor{gray!6}{0.00} & \cellcolor{gray!6}{0.02} & \cellcolor{gray!6}{70.10}\\
circle & 95.19 MB & 0.05 & 0.50 & 83.31 MB & 0.08 & 0.04 & 0.42 & 1.02\\
\cellcolor{gray!6}{torus} & \cellcolor{gray!6}{725.54 MB} & \cellcolor{gray!6}{0.41} & \cellcolor{gray!6}{4.35} & \cellcolor{gray!6}{7.81 MB} & \cellcolor{gray!6}{0.10} & \cellcolor{gray!6}{0.00} & \cellcolor{gray!6}{0.03} & \cellcolor{gray!6}{36.62}\\
swiss-roll & 703.13 MB & 0.40 & 4.23 & 7.98 MB & 0.09 & 0.00 & 0.04 & 35.62\\
\bottomrule
\end{tabular}
\end{table*}

\section{Conclusion}
We have shown that with the right design choices,
the concept of filtration-domination can lead to fast computation
of the topological properties of bifiltered clique complexes.
The presented results already demonstrate this for a natural and important
class of datasets. However, our approach can also be adapted in several directions,
and we plan to investigate them in the full version of the paper:

First of all, we can lift the restriction to 
$1$-critical bifiltered graphs, and allow to associate 
an arbitrary (finite) number of critical values to every edge.
Algorithmically, we can either try to adapt our algorithms to handle such
multicritical edges, or we just pretend that there are multiple copies of the same edge
coming in at different grades, and we can remove them independently. 
The question is which of the option is more efficient.
Such a generalization would allow us to extend our experimental evaluation to
other natural types of bifiltrations, for instance 
\emph{degree-Rips bifiltrations}~\cite{blumbergStability2ParameterPersistent2022,lesnickInteractiveVisualization2D2015,rolle-degree}.

The concepts of filtration-domination and strongly
filtration-domination extend to any number of parameters.
This is also true for
the ideas behind our algorithms. In particular, it is straightforward to adapt
the algorithm to check for strong filtration-domination to such more general
cases. The extension to an efficient filtration-domination test requires
some care, as the data structure for the domination check gets more involved in this case.

\bibliography{refs}

\begin{thebibliography}{10}

\bibitem{aa-circle}
Micha\l{} Adamaszek and Henry Adams.
\newblock The {V}ietoris-{R}ips complexes of a circle.
\newblock {\em Pacific Journal of Mathematics}, 290(1):1--40, 2017.
\newblock \href {https://doi.org/10.2140/pjm.2017.290.1}
  {\path{doi:10.2140/pjm.2017.290.1}}.

\bibitem{Aktas2019}
Mehmet~E. Aktas, Esra Akbas, and Ahmed~El Fatmaoui.
\newblock Persistence homology of networks: methods and applications.
\newblock {\em Applied Network Science}, 4(1), August 2019.
\newblock \href {https://doi.org/10.1007/s41109-019-0179-3}
  {\path{doi:10.1007/s41109-019-0179-3}}.

\bibitem{alliliReducingComplexesMultidimensional2017}
Madjid Allili, Tomasz Kaczynski, and Claudia Landi.
\newblock Reducing complexes in multidimensional persistent homology theory.
\newblock {\em Journal of Symbolic Computation}, 78:61--75, January 2017.
\newblock \href {https://doi.org/10.1016/j.jsc.2015.11.020}
  {\path{doi:10.1016/j.jsc.2015.11.020}}.

\bibitem{alliliAcyclicPartialMatchings2019}
Madjid Allili, Tomasz Kaczynski, Claudia Landi, and Filippo Masoni.
\newblock Acyclic {{Partial Matchings}} for {{Multidimensional Persistence}}:
  {{Algorithm}} and {{Combinatorial Interpretation}}.
\newblock {\em Journal of Mathematical Imaging and Vision}, 61(2):174--192,
  February 2019.
\newblock \href {https://doi.org/10.1007/s10851-018-0843-8}
  {\path{doi:10.1007/s10851-018-0843-8}}.

\bibitem{AttaliLieutierExtendedCol}
Dominique Attali, Andr{\'e} Lieutier, and David Salinas.
\newblock Vietoris\textendash{{Rips}} complexes also provide topologically
  correct reconstructions of sampled shapes.
\newblock {\em Computational Geometry}, 46(4):448--465, May 2013.
\newblock \href {https://doi.org/10.1016/j.comgeo.2012.02.009}
  {\path{doi:10.1016/j.comgeo.2012.02.009}}.

\bibitem{bleherTopologicalDataAnalysis2022}
Michael Bleher, Lukas Hahn, Juan~Angel {Pati{\~n}o-Galindo}, Mathieu Carriere,
  Ulrich Bauer, Raul Rabadan, and Andreas Ott.
\newblock Topological data analysis identifies emerging adaptive mutations in
  {{SARS-CoV-2}}, February 2022.
\newblock \href {http://arxiv.org/abs/2106.07292} {\path{arXiv:2106.07292}}.

\bibitem{blumbergStability2ParameterPersistent2022}
Andrew~J. Blumberg and Michael Lesnick.
\newblock Stability of 2-{{Parameter Persistent Homology}}.
\newblock {\em Foundations of Computational Mathematics}, October 2022.
\newblock \href {https://doi.org/10.1007/s10208-022-09576-6}
  {\path{doi:10.1007/s10208-022-09576-6}}.

\bibitem{bmt-topological}
Omer Bobrowski, Sayan Mukherjee, and Jonathan~E. Taylor.
\newblock {Topological consistency via kernel estimation}.
\newblock {\em Bernoulli}, 23(1):288--328, 2017.
\newblock \href {https://doi.org/10.3150/15-BEJ744}
  {\path{doi:10.3150/15-BEJ744}}.

\bibitem{boissonnatEdgeCollapsePersistence}
Jean-Daniel Boissonnat and Siddharth Pritam.
\newblock {Edge Collapse and Persistence of Flag Complexes}.
\newblock In Sergio Cabello and Danny~Z. Chen, editors, {\em 36th International
  Symposium on Computational Geometry (SoCG 2020)}, volume 164 of {\em Leibniz
  International Proceedings in Informatics (LIPIcs)}, pages 19:1--19:15,
  Dagstuhl, Germany, 2020. Schloss Dagstuhl--Leibniz-Zentrum f{\"u}r
  Informatik.
\newblock \href {https://doi.org/10.4230/LIPIcs.SoCG.2020.19}
  {\path{doi:10.4230/LIPIcs.SoCG.2020.19}}.

\bibitem{botnanIntroductionMultiparameterPersistence2022}
Magnus~Bakke Botnan and Michael Lesnick.
\newblock An {{Introduction}} to {{Multiparameter Persistence}}, March 2022.
\newblock \href {http://arxiv.org/abs/2203.14289} {\path{arXiv:2203.14289}}.

\bibitem{buchetTopAnOutliers2015}
Micka{\"e}l Buchet, Fr{\'e}d{\'e}ric Chazal, Tamal~K. Dey, Fengtao Fan,
  Steve~Y. Oudot, and Yusu Wang.
\newblock {Topological Analysis of Scalar Fields with Outliers}.
\newblock In Lars Arge and J{\'a}nos Pach, editors, {\em 31st International
  Symposium on Computational Geometry (SoCG 2015)}, volume~34 of {\em Leibniz
  International Proceedings in Informatics (LIPIcs)}, pages 827--841, Dagstuhl,
  Germany, 2015. Schloss Dagstuhl--Leibniz-Zentrum f{\"u}r Informatik.
\newblock \href {https://doi.org/10.4230/LIPIcs.SOCG.2015.827}
  {\path{doi:10.4230/LIPIcs.SOCG.2015.827}}.

\bibitem{carlssonPersHom2020}
Gunnar Carlsson.
\newblock Persistent homology and applied homotopy theory.
\newblock In {\em Handbook of homotopy theory}, CRC Press/Chapman Hall Handb.
  Math. Ser., pages 297--330. CRC Press, Boca Raton, FL, 2020.

\bibitem{carlssonCompMultiOneCrit2010}
Gunnar Carlsson, Gurjeet Singh, and Afra Zomorodian.
\newblock Computing multidimensional persistence.
\newblock {\em Journal of Computational Geometry}, 1(1):72--100, 2010.
\newblock \href {https://doi.org/10.20382/jocg.v1i1a6}
  {\path{doi:10.20382/jocg.v1i1a6}}.

\bibitem{carlssonTheoryMultidimensionalPersistence2009}
Gunnar Carlsson and Afra Zomorodian.
\newblock {The Theory of Multidimensional Persistence}.
\newblock {\em Discrete {\&} Computational Geometry}, 42(1):71--93, April 2009.
\newblock \href {https://doi.org/10.1007/s00454-009-9176-0}
  {\path{doi:10.1007/s00454-009-9176-0}}.

\bibitem{chazalPersistenceStabilityGeometric2014}
Fr\'{e}d\'{e}ric Chazal, Vin de~Silva, and Steve Oudot.
\newblock Persistence stability for geometric complexes.
\newblock {\em Geometriae Dedicata}, 173:193--214, 2014.
\newblock \href {https://doi.org/10.1007/s10711-013-9937-z}
  {\path{doi:10.1007/s10711-013-9937-z}}.

\bibitem{cgos-jacm}
Fr\'{e}d\'{e}ric Chazal, Leonidas~J. Guibas, Steve~Y. Oudot, and Primoz Skraba.
\newblock Persistence-based clustering in {R}iemannian manifolds.
\newblock {\em Journal of the ACM}, 60(6):Art. 41, 38, 2013.
\newblock \href {https://doi.org/10.1145/2535927} {\path{doi:10.1145/2535927}}.

\bibitem{dkm-computing}
Tamal~K. Dey, Woojin Kim, and Facundo M\'{e}moli.
\newblock {Computing Generalized Rank Invariant for 2-Parameter Persistence
  Modules via Zigzag Persistence and Its Applications}.
\newblock In Xavier Goaoc and Michael Kerber, editors, {\em 38th International
  Symposium on Computational Geometry (SoCG 2022)}, volume 224 of {\em Leibniz
  International Proceedings in Informatics (LIPIcs)}, pages 34:1--34:17,
  Dagstuhl, Germany, 2022. Schloss Dagstuhl -- Leibniz-Zentrum f{\"u}r
  Informatik.
\newblock \href {https://doi.org/10.4230/LIPIcs.SoCG.2022.34}
  {\path{doi:10.4230/LIPIcs.SoCG.2022.34}}.

\bibitem{deyCompTop2022}
Tamal~K. Dey and Yusu Wang.
\newblock {\em Computational topology for data analysis}.
\newblock Cambridge University Press, Cambridge, 2022.

\bibitem{dx-generalized}
Tamal~K. Dey and Cheng Xin.
\newblock Generalized persistence algorithm for decomposing multiparameter
  persistence modules.
\newblock {\em Journal of Applied and Computational Topology}, February 2022.
\newblock \href {https://doi.org/10.1007/s41468-022-00087-5}
  {\path{doi:10.1007/s41468-022-00087-5}}.

\bibitem{edelsbrunnerCompTop2010}
Herbert Edelsbrunner and John~L. Harer.
\newblock {\em {Computational Topology: An Introduction}}.
\newblock American Mathematical Society, 2010.
\newblock \href {https://doi.org/10.1090/mbk/069} {\path{doi:10.1090/mbk/069}}.

\bibitem{fk-chunk}
Ulderico Fugacci and Michael Kerber.
\newblock {Chunk Reduction for Multi-Parameter Persistent Homology}.
\newblock In Gill Barequet and Yusu Wang, editors, {\em 35th International
  Symposium on Computational Geometry (SoCG 2019)}, volume 129 of {\em Leibniz
  International Proceedings in Informatics (LIPIcs)}, pages 37:1--37:14,
  Dagstuhl, Germany, 2019. Schloss Dagstuhl--Leibniz-Zentrum f{\"u}r
  Informatik.
\newblock \href {https://doi.org/10.4230/LIPIcs.SoCG.2019.37}
  {\path{doi:10.4230/LIPIcs.SoCG.2019.37}}.

\bibitem{fkr-compression}
Ulderico Fugacci, Michael Kerber, and Alexander Rolle.
\newblock Compression for 2-parameter persistent homology.
\newblock {\em Computational Geometry}, 109:101940, February 2023.
\newblock \href {https://doi.org/10.1016/j.comgeo.2022.101940}
  {\path{doi:10.1016/j.comgeo.2022.101940}}.

\bibitem{glisseSwapShiftTrim2022}
Marc Glisse and Siddharth Pritam.
\newblock {Swap, Shift and Trim to Edge Collapse a Filtration}.
\newblock In Xavier Goaoc and Michael Kerber, editors, {\em 38th International
  Symposium on Computational Geometry (SoCG 2022)}, volume 224 of {\em Leibniz
  International Proceedings in Informatics (LIPIcs)}, pages 44:1--44:15,
  Dagstuhl, Germany, 2022. Schloss Dagstuhl -- Leibniz-Zentrum f{\"u}r
  Informatik.
\newblock \href {https://doi.org/10.4230/LIPIcs.SoCG.2022.44}
  {\path{doi:10.4230/LIPIcs.SoCG.2022.44}}.

\bibitem{gudhi}
The {GUDHI project}.
\newblock {\em {GUDHI}: Geometry understanding in higher dimensions}.
\newblock URL: \url{http://gudhi.inria.fr/}.

\bibitem{kn-efficient}
Michael Kerber and Arnur Nigmetov.
\newblock Efficient approximation of the matching distance for 2-parameter
  persistence.
\newblock In Sergio Cabello and Danny~Z. Chen, editors, {\em 36th International
  Symposium on Computational Geometry, SoCG 2020, June 23-26, 2020,
  Z{\"{u}}rich, Switzerland}, volume 164 of {\em LIPIcs}, pages 53:1--53:16.
  Schloss Dagstuhl - Leibniz-Zentrum f{\"{u}}r Informatik, 2020.
\newblock \href {https://doi.org/10.4230/LIPIcs.SoCG.2020.53}
  {\path{doi:10.4230/LIPIcs.SoCG.2020.53}}.

\bibitem{kerberFastMinimalPresentations2020}
Michael Kerber and Alexander Rolle.
\newblock Fast minimal presentations of bi-graded persistence modules.
\newblock In {\em 2021 Proceedings of the Workshop on Algorithm Engineering and
  Experiments ({ALENEX})}, pages 207--220. Society for Industrial and Applied
  Mathematics, January 2021.
\newblock \href {https://doi.org/10.1137/1.9781611976472.16}
  {\path{doi:10.1137/1.9781611976472.16}}.

\bibitem{lesnickInteractiveVisualization2D2015}
Michael Lesnick and Matthew Wright.
\newblock Interactive {{Visualization}} of 2-{{D Persistence Modules}},
  December 2015.
\newblock \href {http://arxiv.org/abs/1512.00180} {\path{arXiv:1512.00180}}.

\bibitem{lesnickComputingMinimalPresentations2020}
Michael Lesnick and Matthew Wright.
\newblock Computing {{Minimal Presentations}} and {{Bigraded Betti Numbers}} of
  2-{{Parameter Persistent Homology}}, February 2022.
\newblock \href {http://arxiv.org/abs/1902.05708} {\path{arXiv:1902.05708}}.

\bibitem{Lo2018}
Derek Lo and Briton Park.
\newblock {Modeling the spread of the Zika virus using topological data
  analysis}.
\newblock {\em {PLOS} {ONE}}, 13(2):e0192120, February 2018.
\newblock \href {https://doi.org/10.1371/journal.pone.0192120}
  {\path{doi:10.1371/journal.pone.0192120}}.

\bibitem{otterRoadmapComputationPersistent2017}
Nina Otter, Mason~A. Porter, Ulrike Tillmann, Peter Grindrod, and Heather~A.
  Harrington.
\newblock A roadmap for the computation of persistent homology.
\newblock {\em {EPJ} Data Science}, 6(1), August 2017.
\newblock \href {https://doi.org/10.1140/epjds/s13688-017-0109-5}
  {\path{doi:10.1140/epjds/s13688-017-0109-5}}.

\bibitem{oudotPersistenceTheory2015}
Steve~Y. Oudot.
\newblock {\em Persistence theory: from quiver representations to data
  analysis}, volume 209 of {\em Mathematical Surveys and Monographs}.
\newblock American Mathematical Society, 2015.
\newblock \href {https://doi.org/10.1090/surv/209}
  {\path{doi:10.1090/surv/209}}.

\bibitem{rolle-degree}
Alexander Rolle.
\newblock {The Degree-Rips Complexes of an Annulus with Outliers}.
\newblock In Xavier Goaoc and Michael Kerber, editors, {\em 38th International
  Symposium on Computational Geometry (SoCG 2022)}, volume 224 of {\em Leibniz
  International Proceedings in Informatics (LIPIcs)}, pages 58:1--58:14,
  Dagstuhl, Germany, 2022. Schloss Dagstuhl -- Leibniz-Zentrum f{\"u}r
  Informatik.
\newblock \href {https://doi.org/10.4230/LIPIcs.SoCG.2022.58}
  {\path{doi:10.4230/LIPIcs.SoCG.2022.58}}.

\bibitem{rs-stable}
Alexander Rolle and Luis Scoccola.
\newblock Stable and consistent density-based clustering, July 2021.
\newblock \href {http://arxiv.org/abs/2005.09048} {\path{arXiv:2005.09048}}.

\bibitem{roweisNonlinearDimensionalityReduction2000}
Sam~T. Roweis and Lawrence~K. Saul.
\newblock Nonlinear dimensionality reduction by locally linear embedding.
\newblock {\em Science}, 290(5500):2323--2326, December 2000.
\newblock \href {https://doi.org/10.1126/science.290.5500.2323}
  {\path{doi:10.1126/science.290.5500.2323}}.

\bibitem{scaramucciaComputingMultiparameterPersistent2020}
Sara Scaramuccia, Federico Iuricich, Leila De~Floriani, and Claudia Landi.
\newblock Computing multiparameter persistent homology through a discrete
  {{Morse-based}} approach.
\newblock {\em Computational Geometry}, 89:101623, August 2020.
\newblock \href {https://doi.org/10.1016/j.comgeo.2020.101623}
  {\path{doi:10.1016/j.comgeo.2020.101623}}.

\bibitem{sheehyLinearSizeApproximationsVietoris2013}
Donald~R. Sheehy.
\newblock Linear-size approximations to the {V}ietoris-{R}ips filtration.
\newblock {\em Discrete \& Computational Geometry}, 49(4):778--796, 2013.
\newblock \href {https://doi.org/10.1007/s00454-013-9513-1}
  {\path{doi:10.1007/s00454-013-9513-1}}.

\bibitem{silvermanDensityEstimationStatistics1986}
Bernard~W. Silverman.
\newblock {\em Density estimation for statistics and data analysis}.
\newblock Monographs on Statistics and Applied Probability. Chapman \& Hall,
  1986.
\newblock \href {https://doi.org/10.1007/978-1-4899-3324-9}
  {\path{doi:10.1007/978-1-4899-3324-9}}.

\bibitem{Welker}
Volkmar Welker.
\newblock Constructions preserving evasiveness and collapsibility.
\newblock {\em Discrete Mathematics}, 207(1-3):243--255, September 1999.
\newblock \href {https://doi.org/10.1016/S0012-365X(99)00049-7}
  {\path{doi:10.1016/S0012-365X(99)00049-7}}.

\end{thebibliography}

\end{document}